\documentclass{article}

\usepackage{graphicx,amsmath,amsfonts,lineno,amsthm}
\usepackage[usenames,dvipsnames]{xcolor}

\newtheorem{theorem}{Theorem}[section]

\newtheorem{lemma}[theorem]{Lemma}

\newtheorem{conjecture}[theorem]{Conjecture}

\newcommand\blfootnote[1]{%
  \begingroup
  \renewcommand\thefootnote{}\footnote{#1}%
  \addtocounter{footnote}{-1}%
  \endgroup
}

\title{A Note on Empty Balanced Tetrahedra in Two colored Point sets in $\mathbb{R}^3$}

\author{Jos\'e M. D\'iaz-Ba\~nez\thanks{Department of Applied Mathematics II, University of Seville, Camino de los Descubrimientos s/n, Seville 41092, Spain. \tt{dbanez@us.es}} 
\and Ruy Fabila-Monroy\thanks{Departamento de Matem\'aticas, CINVESTAV. CDMX, Mexico. Partially supported by Conacyt of Mexico, Grant 253261. \tt{ruyfabila@math.cinvestav.edu.mx}} 
\and Jorge Urrutia \thanks{Instituto de Matem\'aticas, Universidad Nacional Aut\'onoma de M\'exico (UNAM), CDMX, Mexico. \tt{urrutia@matem.unam.mx}. }}

\begin{document}

\maketitle
\blfootnote{\begin{minipage}[l]{0.3\textwidth} \includegraphics[trim=10cm 6cm 10cm 5cm,clip,scale=0.15]{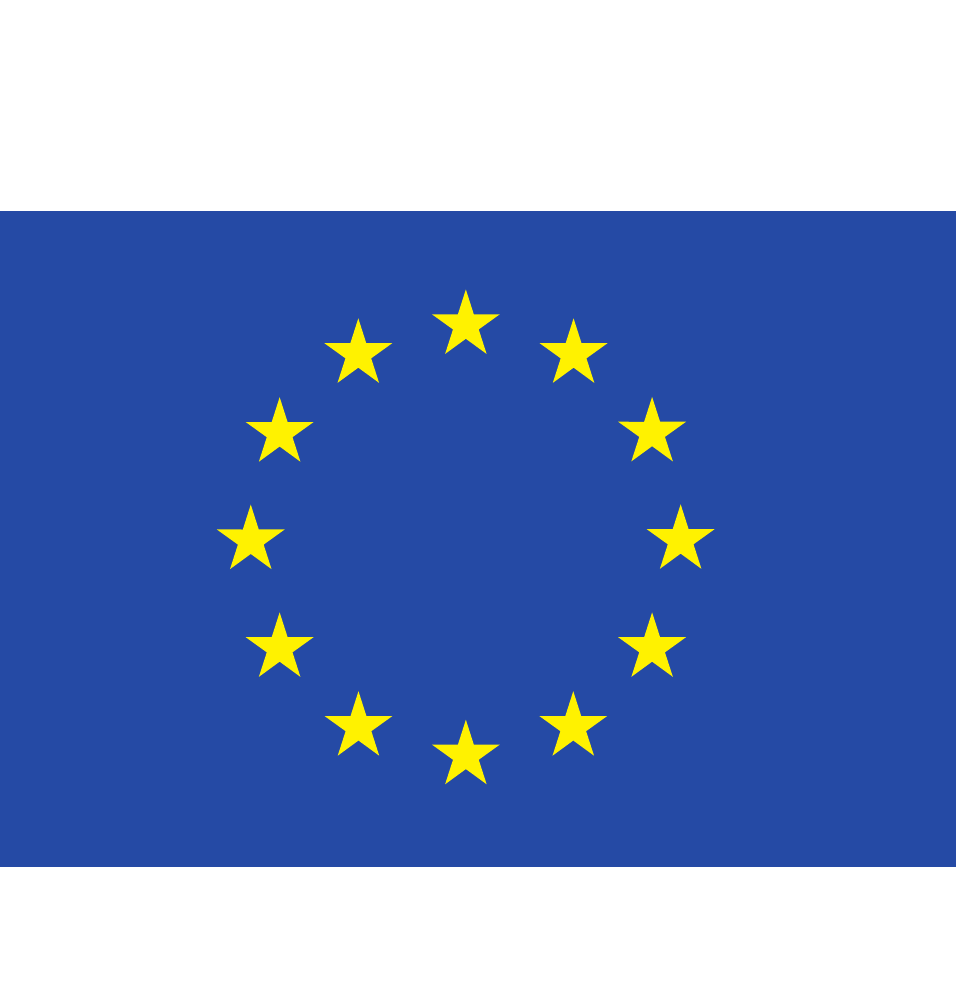} \end{minipage}  \hspace{-2.cm} \begin{minipage}[l][1cm]{0.82\textwidth}
 	  This project has received funding from the European Union's Horizon 2020 research and innovation programme under the Marie Sk\l{}odowska-Curie grant agreement No 734922.
 	\end{minipage}}

\begin{abstract}

Let $S$ be a set of $n$ red and $n$ blue points in general position in $\mathbb{R}^3$. Let 
$\tau$ be a tetrahedra with vertices on $S$. We say that $\tau$ is \emph{empty}
if it does not contain any point of $S$ in its interior. We say that 
$\tau$ is \emph{balanced} if it contains two blue vertices and two red vertices. 
In this paper we show that $S$ spans $\Omega(n^{5/2})$ empty balanced tetrahedra. 
\end{abstract}

\section{Introduction}

The study of \emph{holes} in point sets in $\mathbb{R}^2$ and $\mathbb{R}^3$
was started many years ago in a seminal papers of Paul Erd\H{o}s and George Szekeres~\cite{ErdosSek,erdos1978some}. A list of some of the papers studying 
the existence of holes in point sets, as well as variants on colored point sets
is given at the end of our paper.
\nocite{gulik_4_gon,friedman_4_gon,brass_4_gon,clemens_4_gon,balanced_6_holes,almost_quad,balanced_4_holes,non_convex_quad,
mono_us,mono_pach,almost,variants, empty_simp,onalmost,nicolas2007empty,gerken2008empty}

All point sets $S$ considered in this paper will be assumed to be in general position, and to have $n$ red, and $n$ blue points. We will call them balanced point sets.

A balanced $4$-hole of a balanced point set $S$ on the plane is a polygon $P$, not necessarily convex, whose vertices are in $S$, and having two red, and two blue vertices; in addition $P$ contains no extra
elements of $S$ on its interior.
Balanced $4$-holes for point sets on the plane were studied in~\cite{balanced_4_holes}.
They proved that any balanced point set on the plane always has a quadratic number
of not necessarily convex balanced $4$-holes, and that there are balanced point sets with no
\emph{convex} $4$-holes. In this paper we extend these results to $\mathbb{R}^3$.

Let $S$ be a balanced point set in general position in $\mathbb{R}^3$. Let 
$\tau$ be a tetrahedra with vertices on $S$. We say that $\tau$ is \emph{empty}
if it does not contain any point of $S$ in its interior. We say that 
$\tau$ is \emph{balanced} if it contains two blue and two red vertices. 
In this paper we show that $S$ spans $\Omega(n^{5/2})$ empty balanced tetrahedra. 

To show our lower bound we consider a similar problem in the plane. Suppose that
$S$ is now a set of $n$ red and at least $n-1$ blue points in general position
in the plane. We say that a triangle $\tau$ with vertices in $S$ is a \emph{red}
$(2,1)$-triangle if two of its vertices are red and one of them is blue. Blue 
$(2,1)$-triangles are defined in a similar way.
We show that $S$ spans $\Omega(n^{3/2})$ empty red $(2,1)$-triangles. 

\section{Lower Bounds}
\begin{lemma}\label{lem:triangles}
 Every set of $S$ of $n$ red and at least $n-1$ blue points in general position in the plane
 spans $\Omega(n^{3/2})$ empty red $(2,1)$-triangles.
\end{lemma}
\begin{proof}
 Let $p$ be a red point of $S$. Sort the points of $S\setminus \{p\}$ clockwise
 by angle around $p$. Let $I_1,\dots,I_m$ be the maximal intervals of consecutive
 red points of $S \setminus \{p\}$ in this order.
 
 Suppose that $m \le \sqrt{n}$. 
 For every $1 \le i \le m$ let $J_i$ be a subinterval of $I_i$ of 
at least $|I_i|/2$ points such that $p \notin \operatorname{Conv}(J_i)$.
Note that $\operatorname{Conv}(J_i) \cap S = J_i$.

Let $q_1,\dots,q_{n'}$ be the blue points in $S$ sorted by distance to $\operatorname{Conv}(J_i)$.
For every $1 \le j \le |I_i|-1 $ let  $P_j:=\operatorname{Conv}(J_i\cup \{q_1,\dots,q_j\}) \cap S$.
Sort the points of $P_j \setminus \{q_j\}$ clockwise by angle around $q_j$. 
Note that every pair of a red point in $S$ and a red point in $J_i$, consecutive in this order, defines an empty red $(2,1)$-triangle.
This triangle contains $q_j$ and a point in $J_i$ as its vertices. Since there are $j-1$ blue points different from $q_j$ in $P_i$ , there are at least $|J_i|-j$ such triangles. 
Therefore, $S$ spans at least \[\sum_{j=1}^{|J_i|-1} |J_i|-j=\frac{|J_i|(|J_i|-1)}{2}\ge \frac{|I_i|^2}{8}-\frac{|I_i|}{4}.\]
red $(2,1)$-triangles containing a vertex of $I_i$. By considering all such triangles for every $I_i$ we overcount each triangle
at most two times. We have at least
\[\frac{1}{2}\sum_{i=1}^{m} \left (\frac{|I_i|^2}{8}-\frac{|I_i|}{4} \right ) \ge \frac{(n-1)^2}{16m}-\frac{(n-1)}{8}\ge \frac{1}{32}n^{\frac{3}{2}},\]
empty red $(2,1)$-triangles. The first inequality comes from the fact that this number is minimized when all $I_i$ have the same cardinality;
the last inequality holds for $n \ge 22$. 

Suppose that $m > \sqrt{n}$. Every endpoint of an $I_i$, with the exception of at most two, has a consecutive blue point. These pair
and $p$ define an empty red $(2,1)$-triangle. Thus we have at least \[2\sqrt{m}-2 \ge 2\sqrt{n},\]
 empty red $(2,1)$-triangles containing $p$ as a vertex. 

By repeating this argument for every red point either  we get a total $\frac{1}{32}n^{\frac{3}{2}}$ empty red $(2,1)$-triangles, or for every
red point $p$ we get $2\sqrt{n}$ empty red $(2,1)$-triangles containing $p$  as a vertex. The result follows. 
\end{proof}

\begin{theorem}
 Let $S$ be a set of $n$ red and $n$ blue points in general position in $\mathbb{R}^3$.
 Then $S$ spans $\Omega(n^{5/2})$ empty balanced tetrahedra.
\end{theorem}
\begin{proof}
For every blue point $p$ of $S$ we do the following. Let $\Pi$ be a plane
containing $p$ with the following properties. $\Pi$ contains no other point
of $S$;
at least $\lfloor n/2\rfloor$ of the red points of $S$
lie above $\Pi$ and at least $\lfloor n/2\rfloor$ of the red points of $S$
lie below $\Pi$. 
Without loss of generality assume that  above $\Pi$ there are at least $\lfloor n/2 \rfloor$ blue points of $S$.
Let $\Pi'$ be a plane parallel to $\Pi$ and above $p$. Let $S'$ be the points of $S$ above
$\Pi$. Let $q$ be a point in $S'$ above $\Pi$; let $r$ be the
infinite ray with apex $p$ and passing through $q$. Project $q$ to the point $r \cap \Pi'$. 
Let $S''$ be the image of these projections. By Theorem~\ref{lem:triangles}
$S''$ spans $\Omega(n^{3/2})$ empty red $(2,1)$-triangles. The preimages of the vertices
of these triangles together with $p$ define an empty balanced tetrahedra of $S$.
Since each such tetrahedra is counted at most twice we have at $\Omega(n^{5/2})$
empty balanced tetrahedra spanned by $S$. 
\end{proof}

We close our paper with two conjectures:

\begin{conjecture}
Let $S$ be a set of $n$ red and $n$ blue points in general position in $\mathbb{R}^3$
spans a cubic number of balanced tetrahedra.
\end{conjecture}

\begin{conjecture}
Let $S$ be a set of $n$ red and $n$ blue points in general position in $\mathbb{R}^2$
has a quadratic number of red $(2,1)$-triangles.
\end{conjecture}

We observe that it is easy to see that number of red $(2,1)$-triangles plus the number of
blue $(2,1)$-triangles is quadratic. Any line passing trough a blue and a red point of $S$, when
moved up or down generates a red or a blue triangle when it hits first a point in $S$.

\textbf{Acknowlodegments}
This work was initiated at the VII Spanish Workshop on Geometric Optimization, El Roc\'io,
Huelva, Spain, held June 20-24, 2016. We thank the other
participants of that workshop---
Sergey Bereg, Kirk Boyer,
Evaristo Caraballo,
David Flores,
Marco A. Heredia,
Paul Horn,
Nadine Kroher,
Mario A. L\'opez,
Pablo P\'erez-Lantero,
Adriana Ramírez Vigueras,
Max Rotschke and
Inmaculada Ventura---for helpful discussions
and contributing to a fun and creative atmosphere.

\small \bibliographystyle{alpha} \bibliography{balanced}

\end{document}